%% file: ms.tex
\documentclass[11pt]{article}

\usepackage[utf8]{inputenc}
\usepackage[T1]{fontenc}
\usepackage{booktabs}
\usepackage{amsfonts}
\usepackage{microtype}
\usepackage[dvipsnames]{xcolor}
\usepackage{algorithm}
\usepackage{algorithmicx}
\usepackage{algpseudocode}
\usepackage{graphicx}
\usepackage{float}
\usepackage{amsmath}
\usepackage{amsthm}
\usepackage{hyperref}
\usepackage{cleveref}
\usepackage{xurl}
\usepackage{mathtools}
\usepackage{bm}
\usepackage{natbib}
\usepackage[margin=1in]{geometry}
\usepackage{tikz}
\usetikzlibrary{shapes, shapes.symbols, backgrounds, matrix, calc, arrows, math, arrows.meta, decorations.pathreplacing, decorations.markings}
\newtheorem{theorem}{Theorem}[section]
\newtheorem{lemma}[theorem]{Lemma}

\newtheorem{remark}[theorem]{Remark}

\newtheorem{example}[theorem]{Example}
\usepackage{mathrsfs}
\newcommand{\E}{\operatorname*{\mathbb E}}
\newcommand{\argmax}{\operatorname*{argmax}}

\usepackage{crossreftools}
\renewcommand{\tilde}{\widetilde}
\allowdisplaybreaks

\title{No-regret Learning in Repeated First-Price Auctions with Budget Constraints}

\author{%
    Rui Ai \thanks{Equal contribution} \\ Peking University \\ \texttt{ruiai@pku.edu.cn} \and
    Chang Wang \footnotemark[1] \\ Peking University \\ \texttt{wchang@pku.edu.cn} \and
    Chenchen Li \\ JD.COM \\ \texttt{lcc104@qq.com} \and
    Jinshan Zhang \\ Zhejiang University \\ \texttt{zhangjinshan@zju.edu.cn} \and
    Wenhan Huang \\ Shanghai Jiaotong University \\ \texttt{rowdark@sjtu.edu.cn} \and
    Xiaotie Deng \thanks{Corresponding author} \\ Peking University \\ \texttt{xiaotie@pku.edu.cn}
}

\begin{document}

\maketitle

\begin{abstract}
    Recently the online advertising market has exhibited a gradual shift from second-price auctions to first-price auctions. Although there has been a line of works concerning online bidding strategies in first-price auctions, it still remains open how to handle budget constraints in the problem. In the present paper, we initiate the study for a buyer with budgets to learn online bidding strategies in repeated first-price auctions. We propose an RL-based bidding algorithm against the optimal non-anticipating strategy under stationary competition. Our algorithm obtains $\tilde O(\sqrt T)$-regret if the bids are all revealed at the end of each round. With the restriction that the buyer only sees the winning bid after each round, our modified algorithm obtains $\tilde O(T^{\frac{7}{12}})$-regret by techniques developed from survival analysis. Our analysis extends to the more general scenario where the buyer has any bounded instantaneous utility function with regrets of the same order.
\end{abstract}

\section{Introduction}
The market for online advertising has extensively grown in recent years. It is estimated that the volume of online advertising worldwide will reach 500 billion dollars in 2022 \citep{spending2021}. In such a market, advertising platforms use auctions to allocate ad opportunities. Typically, the advertiser has a limited amount of capital for an advertisement campaign and thus all rounds of competition are \emph{interconnected by budgets}. Furthermore, the advertiser has very limited knowledge of 1) her valuation of certain keywords and 2) the competition she is facing. Many works have been devoted to developing algorithms for \emph{learning} strategies for optimally spending the budget in repeated \emph{second}-price auctions (see \Cref{subsec:related_work}).

Nevertheless, we have witnessed numerous switches from second-price auctions to first-price auctions in the online advertising market. A recent remarkable example is Google AdSenses' integrated move at the end of 2021 \citep{google}. Earlier examples also include AppNexus, Index Exchange, and OpenX \citep{sluis2017big}. This industry-wide shift is due to various factors including a fairer transactional
process and increased transparency. Further, the shift to first-price auctions brings about major importance to the following open question which is barely considered in previous works:
\begin{center}
    \itshape How should budget-constrained advertisers learn to compete in repeated \emph{first-price} auctions?
\end{center}
This paper thus initiates the study of learning to bid with budget constraints in repeated first-price auctions. We note that the application of first-price auctions with budgets is not limited to online advertising mentioned above. Traditional competitive environments like mussel trade in Netherlands \citep{van2001sealed}, modern price competition, and procurement auctions (e.g. U.S. Treasury Securities auction \citep{chari1992us}) are examples as well.

\paragraph{Challenges and contributions}
The challenges in this setting are two-fold.

The first challenge relates to information structure. In practical auctions, it is often the case that only the highest bid is revealed \citep{esponda2008information}. This is known as \emph{censored}-feedback or \emph{winner's curse} in literature \citep{capen1971competitive}. This affects the information structure of learning since the buyer learns less information when she wins. This makes the problem challenging compared to standard contextual bandits (c.f. \Cref{subsec:related_work}).

The second challenge is more fundamental. It is known that the strategy in first-price auctions is notoriously complex to analyze, even in the static case \citep{lebrun1996existence}. To get an intuitive feeling of this difficulty in our problem compared to repeated second-price auctions, let us consider the offline case where the opponents' bids are all known. Given the budget, the problem for second-price auctions can be reduced to a pure knapsack problem, where the budget is regarded as weight capacity and prices as weights. This structure enables mature techniques including duality theory to be applied to study the benchmark strategy. Pitifully in first-price auctions, since the payment depends on the buyer's own bid, the previous approach/benchmark is not directly usable. We provide a concrete example to further illustrate such difficulty.
\begin{example}\label{exa:benchmark}
    Consider a case where the buyer's value $v$ follows a uniform distribution on $[0.4, 1]$ and the highest bid $m$ of the opponents' follows a uniform distribution on $[0, 0.5]$. The budget $B$ is set to be half of the time horizon $0.5T$. Now, the first-best benchmark (an anticipating\footnote{An algorithm is anticipating if bid selection depends on future observations, see \citet{flajolet2017real}.} strategy) is
    \begin{align*}
        & \E_{\substack{\bm v \sim F^T \\ \bm m \sim G^T}} \left[\max_{b_1, \dotsc, b_T} \sum_{t=1}^T (v_t-b_t) \mathbf 1_{\{b_t \geq m_t\}} \right]
        \\ & \text{subject to} \quad \sum_{t=1}^T \mathbf 1_{\{b_t \geq m_t\}} b_t \leq B \quad \forall v_1, \dotsc, v_T; m_1, \dotsc, m_T.
    \end{align*}
    In hindsight, we need to pay as little as possible. Using the theory of knapsack, the utility turns out to be $T \cdot \E[\mathbf 1_{\{V \geq M\}}(V-M)]^+=0.45T$. On the contrary, the optimal non-anticipating bidding strategy in a first-price auction is to bid $\frac{V}{2}$ and the utility is $T \cdot \E[\mathbf 1_{\{\frac{V}{2} \geq M\}}\frac{V}{2}]=0.26T$. There is already an $\Omega(T)$ separation between the first-best benchmark and the ideal case with full information.
\end{example}
This example shows that simple characterization of the optimal in previous works is not applicable. Indeed, it remains unclear what methodology can be applied in first-price auctions with budgets (see \citep[\S 2.4]{balseiro2019learning} for further discussions).

The present paper takes the first step to combat the challenges mentioned above with a dynamic programming approach. Correspondingly, our contribution is also two-fold:
\begin{itemize}
    \item We provide an RL-based learning algorithm. Through the characterization of the optimal strategy, we obtain $\tilde O(\sqrt T)$-regret guarantee for the algorithm in the full-feedback case\footnote{This is especially practical in public-sector auctions \citep{chari1992us} as regulations mandate all bids to be revealed.}.
    \item In the censored-feedback setting, by techniques developed from survival analysis, we modify our algorithm and obtain a regret of $\tilde O(T^{\frac{7}{12}})$.
\end{itemize}

\subsection{Related Work}\label{subsec:related_work}

\paragraph{Repeated second-price auctions with budgets}
There is a flourishing source of literature on bidding strategies in repeated auctions with budgets. Through the lens of online learning, \citet{balseiro2019learning} identify asymptotically optimal online bidding strategies known as \emph{pacing} (a.k.a. bid-shading in literature) in repeated second-price auctions with budgets. Inspired by the pacing strategy, \citet{flajolet2017real} develop no-regret non-anticipating algorithms for learning with contextual information in repeated second-price auctions. Another line of works that uses similar techniques in the present paper includes \cite{amin2012budget, tran2014efficient, gummadi2012repeated}. \citet{gummadi2012repeated} and \citet{amin2012budget} study bidding strategies in repeated second-price auctions with budget constraints, but the former does not involve any learning and the latter does not provide any regret analysis (their estimator is biased). \citet{tran2014efficient} derive regret bounds for the same scenario but the optimization objective is the number of items won instead of value or surplus. \citet{baltaoglu2017online} also use dynamic programming to tackle repeated second-price auctions with budgets. However, they assume per-round budget constraints and their dynamic programming algorithm is for allocating bids among multiple items. Again, we emphasize that no prior work has been done in repeated first-price auctions with budgets, since the structure of the problem (compared to second-price variants) is fundamentally different (recall \Cref{exa:benchmark}).

\paragraph{Repeated first-price auctions without budgets} 
Two notable works concerning repeated first-price auctions are \cite{han2020optimal} and \cite{han2020learning}. In \cite{han2020optimal}, they introduce a new problem called monotone group contextual bandits and then obtain an $O(\sqrt T \ln ^2T)$-regret algorithm for repeated first-price auctions \emph{without} budget constraints under stationary settings. In \cite{han2020learning}, they concentrate on an adversarial setting and develop a mini-max optimal online bidding algorithm with $O(\sqrt T \ln T)$ regret against all Lipschitz bidding strategies. A crucial difference is that in the present paper, budgets are involved thus the bandit algorithm by \citet{han2020optimal} is not suitable for our needs.

\paragraph{Bandit with knapsack} From the bandit side, \citet{badanidiyuru2013bandits} develop bandit algorithms under resource constraints. They show that their algorithm can be used in dynamic procurement, dynamic posted pricing with limited supply, etc.
However, since the bidder observes her value \emph{before} bidding in our problem, results by \citet{badanidiyuru2013bandits} cannot be directly applied to our setting. Our setting also relates to contextual bandit problems with resource constraints \citep{badanidiyuru2014resourceful, agrawal2016linear, agrawal2016efficient}. Nevertheless, applying this contextual bandit approach requires discretizing the action space, which needs Lipschitz continuity of distributions. Our approach does not rely on any continuity assumption. Further, the performance guarantee (typically $\tilde O(T^{\frac{2}{3}})$) is worse than ours. It also does not fit into our information structure when the feedback is censored.

\section{Preliminaries}

\paragraph{Auction mechanism}
We consider a repeated first-price auction with budgets. Specifically, we suppose that the buyer has a limited budget $B$ to spend in a time horizon of $T \leq +\infty$ (can be \emph{unknown} to her) rounds. At the beginning of each round $t=1, 2, \dotsc, T$, the bidder privately observes a value $v_t$ for a fresh copy of item and bids $b_t$ according to her past observations $\bm h_t$ and value $v_t$. Denote the strategy she employs as $\pi \colon (v_t, B_t, \bm h_t) \to b_t$, which maps her current budget $B_t$, value $v_t$ and past history $\bm h_t$ to a bid. Let $m_t$ be the maximum bid of the other bidders. Since the auction is a \emph{first} price auction, if $b_t$ is larger than $m_t$, then the buyer wins the auction, is charged $b_t$, and obtains a utility of $r_t$; otherwise she loses and $r_t=0$. Therefore, the instantaneous utility of the buyer is
\[ r_t=(v_t-b_t) \mathbf 1_{\{b_t \geq m_t\}}. \]

The exact information structure of history the buyer observes is dictated by how the mechanism reveals $m_t$. In full generality, we assume that the feedback is censored, i.e. only the highest bid is revealed at the end of each round and the winner does not observe $m_t$ exactly. This is considered to be an informational version of \emph{winner's curse} \citep{capen1971competitive} and is of practical interest \citep{esponda2008information}. For the purpose of modeling, we suppose that ties are broken in favor of the buyer but this choice is arbitrary and by no means a limitation of our approach.

Next, we state the assumptions on $m_t$ and $v_t$. Without loss of generality, we assume that $b_t, m_t, v_t$ are normalized to be in $[0, 1]$. In the present paper, we consider a stochastic setting where $m_t$ and $v_t$ are drawn from some distributions $F, G$ \emph{unknown} to the buyer, respectively, and independent from the past. We will also refer to the cumulative distribution functions of $F, G$ with the same notations. No further assumptions will be made on $F, G$. Now, the expected instantaneous utility of the buyer at time $t$ with strategy $\pi$ is
\[ R^{\pi}(v_t, b_t) = \E_{m_t \sim F}[r_t]=(v_t-b_t)F(b_t). \]
To argue for the reasonability of this assumption, note that although other buyers may also involve learning behavior, it is typical that in a real advertising market, there are a large number of buyers. The specific buyer only faces a different small portion of them and is completely oblivious of whom she is facing in each round. Therefore, the buyer's sole objective is to maximize her utility (instead of fooling other buyers) and by the law of large numbers, the price $m_t$ and value $v_t$ the buyer observes are independent and identically distributed at least for a period of time.

\paragraph{Buyer's target and regret}
The buyer aims at maximizing her long-term accumulated utility subject to the budget constraints. Recall that the instantaneous utility of the buyer is $r_t=(v_t-b_t) \mathbf 1_{\{b_t \geq m_t\}}$. The payment is $c_t=b_t \mathbf 1_{\{b_t \geq m_t\}}$ and the budget will then decrease accordingly as the payment incurs. She can continue to bid as long as budget has not run out but must stop at
\[ \tau^* = \min \left\{T+1, \min \left\{t \in \mathbb N : \sum_{\tau=1}^t c_{\tau} = B \right\}+1 \right\}. \]
The buyer's problem now becomes to determine how much to bid in each round to maximize her accumulated utility. In line with works \cite{gummadi2012repeated, golrezaei2019dynamic, deng2021prior}, the buyer adopts a discount factor $\lambda \in (0, 1)$. She takes discounts since she does not know $T$ or $\tau^*$. Discount factors are interpreted to be the estimate of the probability that the repeated auction will last for at least $t$ rounds \citep{devanur2014perfect, drutsa2018weakly}. On the economic side, in important real-world markets like online advertising platforms, the buyers are impatient for opportunities since companies of different sizes have different capabilities. Discount factors model how impatient the buyer is \citep{drutsa2017horizon, vanunts2019optimal}.
Now the buyer's optimization problem is to determine a \emph{non-anticipating} strategy $\pi$ for the following:
\begin{align*}
    & \max_{\pi} \quad \E_{\substack{\bm v \sim F^T \\ \bm m \sim G^T}} \left[\sum_{t=1}^T \lambda^{t-1} r_t \right]
    \\ & \text{subject to} \quad \sum_{t=1}^T \mathbf 1_{\{b_t \geq m_t\}} b_t \leq B \quad \forall v_1, \dotsc, v_T; m_1, \dotsc, m_T,
\end{align*}
where $b_t = \pi(v_t, B_t, \bm h_t)$. Here, $\bm v \coloneqq (v_1, \dotsc, v_T)$ denotes the sequence of private values the buyer observes and $\bm m \coloneqq (m_1, \dotsc, m_T)$ is the sequence of the highest bids of the other bidders. $V^{\pi}(B, t)$ denotes the expected accumulated utility using strategy $\pi$ with budget $B$ and starting from time $t$. Let $\pi^*$ denote the optimal bidding strategy with the knowledge of the underlying distributions $F$ and $G$. The corresponding expected accumulated utility is $V^{\pi^*}$.

We now come to define the regret. Given time $T$, the regret is defined to be the difference between the accumulated utility of the buyer's bidding strategy $\pi$ and that of the optimal non-anticipating bidding strategy $\pi^*$
\begin{equation}\label{eqn:def_regret}
    \text{Regret} = \E_{\bm v \sim G^T} \left[\sum_{t=1}^T R^{\pi^*}(v_t, b_t^*)-R^{\pi}(v_t, b_t) \right].
\end{equation}

\section{Bidding Algorithm and Analysis}
In this section, we present our bidding algorithm and the high-level ideas in the analysis of regret. We first consider the case where the feedback is not censored, i.e. the buyer is aware of $m_t$ no matter whether she wins or not. Then we extend our algorithm to the case where the feedback is censored with techniques developed from survival analysis.

\subsection{Full Feedback}\label{subsec:full_feedback}
When $F$ and $G$ are known, the buyer's problem can be viewed as offline. The technical challenge lies in the observation that even when the distributions are known, the buyer's problem cannot be directly analyzed as a knapsack problem. To tackle this challenge, we use a dynamic programming approach to solve the problem. In particular, the optimal strategy $\pi^*$ satisfies the following Bellman equation:
\begin{align*}
    b^*(B_{\tau}, v) &\in \argmax_b [(v-b) + \lambda V(B_{\tau}-b, \tau+1)] F(b) + \lambda V(B_{\tau}, \tau+1)(1 - F(b)), \\
    V(B_{\tau}, \tau) &= \E_v [(v-b^*) + \lambda V(B_{\tau}-b^*, \tau+1)] F(b^*) + \lambda V(B_{\tau}, \tau+1)(1 - F(b^*)),
\end{align*}
for all $\tau \in \mathbb N$ and $0 \leq B_{\tau} \leq B$. Note that for any $B_{\tau}<0$, $V(B_{\tau}, \tau)=-\infty$. By choosing appropriate initialization conditions, we can solve the equation recursively and obtain the optimal bidding strategy together with the function $V(\cdot, \cdot)$. The above recursion also characterizes the optimal solution, which will be used in the analysis later. 

When the buyer does not have the information of $F$ and $G$, she can learn the distributions from past observations. Therefore, it is natural to maintain estimations $\hat{F}$ and $\hat{G}$ of the true distributions. Our algorithm for the full-feedback case is now depicted in \Cref{algo:full_feedback}. To ease technical loads, we first assume the knowledge of $G$ and only estimate $F$ in \Cref{algo:full_feedback}. Later, we will add the estimation of $G$ and its analysis is presented in \Cref{cor:unknownG}.

\begin{algorithm}[htbp]
    \caption{Algorithm for the full-feedback case}
    \label{algo:full_feedback}
    \begin{algorithmic}[1]
        \State \textbf{Input}: Initial budget $B$ and constant $C_1$ \Comment{$C_1$ is an arbitrary positive constant}
        \State Initialize the estimation $\hat{F}$ of $F$ to a uniform distribution over $[0, 1]$ and $B_1 \gets B$ \label{algoline:init_full}
        \For{$t=1, 2, \dotsc$}
            \State Observe the value $v_t$ in round $t$
            \State Let $t_0$ be the smallest integer that satisfies $\lambda^{t_0-t} \frac{1}{1 - \lambda} < \frac{C_1}{\sqrt t}$\label{algoline:init}
            \State Set $V_{\hat F}(B_{t_0}, t_0)=0$ for any $B_{t_0}$ \Comment{$V_{\hat F}$ is algorithm's estimation of $V$}
            \For{$\tau=t_0, t_0-1, \dotsc, t$}\label{algoline:full_dp}
                \State $Q_{v, \hat{F}}(B_{\tau}, \tau, b) \gets [(v-b) + \lambda V_{\hat F}(B_{\tau}-b, \tau+1)] \hat{F}(b) + \lambda V_{\hat F}(B_{\tau}, \tau+1)(1 - \hat{F}(b))$
                \State Solve the optimization problem $\hat b_{\tau}^* \gets \argmax_b Q_{v, \hat{F}}(B_{\tau}, \tau, b)$
                \State $V_{\hat F}(B_{\tau}, \tau) \gets \E_{v \sim G}[Q_{v, \hat{F}}(B_{\tau}, \tau, \hat b_{\tau}^*)]$
            \EndFor \label{algoline:dp_full_end}
            \State Place a bid $\hat b_t \gets \argmax_b Q_{v, \hat{F}}(B_t, t, b)$
            \State Observe $m_t, c_t$ and $r_t$ from this round of auction and update $\hat{F}(x) = \frac{1}{t} \sum_{i=1}^t \mathbf 1_{\{m_i \leq x\}}$. \label{algoline:final_full}
            \State $B_{t+1} \gets B_t-c_t$. If $B_{t+1} \leq 0$ then halt. \label{algoline:final_full_end}
        \EndFor
    \end{algorithmic}
\end{algorithm}

Similar to prior work \citep{amin2012budget}, \Cref{algo:full_feedback} performs exploration and exploitation simultaneously. The buyer initializes the estimation of $F$ to a uniform distribution (\Cref{algoline:init_full}). At round $t$, the buyer observes her valuation. Then, she uses her estimation of $F$ to solve the dynamic programming problem recursively\footnote{For the non-trivial case $B \leq T$, this can be solved in $O \left(\frac{T^{4.5}}{(1 - \lambda)^6} \right)$ time with only $O(T^{-\frac{1}{2}})$ loss \citep[see, e.g.][]{chow1989complexity}.} to obtain an estimation of the optimal bid (\Cref{algoline:full_dp}\textasciitilde\Cref{algoline:dp_full_end}). To provide a base case for recursion, note that for sufficiently large $t_0 \gg t$, $V_{\hat F}(\cdot, t_0)$'s impact to $V_{\hat F}(\cdot, t)$ is negligible due to the discount $\lambda^{t_0-t}$. Therefore, the buyer can approximate $V_{\hat F}(\cdot, t_0)$ with zero for $t_0$ (\Cref{algoline:init}). Finally, the auction proceeds with $m_t, r_t, c_t$ revealed and the buyer updates her information accordingly (\Cref{algoline:final_full}\textasciitilde\Cref{algoline:final_full_end}).

\paragraph{Analysis of regret}
The establishment of the regret bounds will be given in two steps. First, we will show that the buyer's estimation $V_{\hat F}$ is approximately accurate with a sufficient number of samples. This relies on the estimation of $F$. Concentration inequalities are thus intrinsic to our analysis. Second, we bridge the regret defined in \Cref{eqn:def_regret} with $V$ and $V_{\hat F}$. This is done by a series of auxiliary quantities measuring the regret. We obtain the desired result by combining the two steps.

We first present the following lemma \citep{dvoretzky1956asymptotic, massart1990tight}, which states a uniform convergence result for the estimation of cumulative distribution functions.

\begin{lemma}[Dvoretzky--Kiefer--Wolfowitz]\label{lem:DKW}
    Given $t \in \mathbb N$, let $m_1, m_2, \dotsc, m_t$ be real-valued independent and identically distributed random variables with cumulative distribution function $F$. Let $\hat F_t$ denote the associated empirical distribution function defined by $\hat F_t(x)=\frac{1}{t}\sum_{i=1}^t \mathbf{1}_{\{m_i\le x\}}$ where $x\in \mathbb{R}$. Then with probability $1-\delta$, it holds
    \[ \sup_x |\hat F_t(x)-F(x)| \leq \sqrt{\frac{1}{2} \ln \frac{2}{\delta}} t^{-\frac{1}{2}}. \]
\end{lemma}

With \Cref{lem:DKW} in hand, given any $T$, we show a bound for the difference between $V(B_t, t)$ and $V_{\hat F} (B_t, t)$ for any $1\le t\le T$ (recall that $V(B_t, t)$ is the accumulated utility of the optimal strategy with the knowledge of $F$). We prove this using induction. Note that the induction basis is a little tricky due to the base case of recursion in \Cref{algo:full_feedback}. Therefore, in the following lemma, we first deal with the induction step.
\begin{lemma}\label{lem:vv_diff}
    For any round $t\le  T$, budget $B_t$ and with probability at least $1 - \frac{\delta}{T}$, we have the following bounds for the estimated $V_{\hat F}$ and the ground truth with $F$ if $\sup_{B_{t_0}} |V(B_{t_0},t_0)-V_{\hat F}(B_{t_0},t_0)|\le \sqrt{\frac{1}{2} \ln \frac{2T}{\delta}} \frac{1 + \lambda}{(1 - \lambda)^2} t^{-\frac{1}{2}}$:
    \[ |V(B_t, t) - V_{\hat F}(B_t, t)| \leq Ct^{-\frac{1}{2}} = \tilde O \left(\frac{1}{\sqrt t} \right) \quad \text{where} \quad C = \sqrt{\frac{1}{2} \ln \frac{2T}{\delta}} \frac{1 + \lambda}{(1 - \lambda)^2}. \]
\end{lemma}

Note that the lemma above assumes that $\sup_{B_{t_0}} |V(B_{t_0},t_0)-V_{\hat F}(B_{t_0},t_0)|$ is bounded above by $\sqrt{\frac{1}{2} \ln \frac{2T}{\delta}} \frac{1 + \lambda}{(1 - \lambda)^2} t^{-\frac{1}{2}}$, which holds if the base case $V_{\hat F}(B_{t_0},t_0)$ is set accurately (i.e. $V(B_{t_0},t_0)=V_{\hat F}(B_{t_0},t_0)$). We use $\tilde V_{\hat F}(B_t, t)$ to denote the estimated value function using $\hat{F}$ if the base is indeed accurate. In the following lemma, we show that using the alternative initialization method specified in \Cref{algoline:init} of \Cref{algo:full_feedback}, $\sup_{B_t} |V_{\hat F}(B_t, t) - \tilde V_{\hat F}(B_t, t)|$ actually satisfies the condition of \Cref{lem:vv_diff}.

\begin{lemma}\label{lem:tildehat_diff}
    Suppose $\tilde{V}_{\hat F}(B_{t_0},t_0)=V(B_{t_0},t_0)$ and $\tilde{V}_{\hat F}(B_t,t)$ is then computed by the recursive procedure in \Cref{algo:full_feedback}. Then it holds that for any $\tau\le t_0$ and $B_\tau$:
    \[ |V_{\hat F}(B_\tau,\tau) - \tilde{V}_{\hat F}(B_\tau,\tau)|\le \frac{1}{1-\lambda}\lambda^{t_0-\tau}. \]
    In particular, when $\tau=t$, we have $\sup_{B_t} |V_{\hat F}(B_t,t) - \tilde{V}_{\hat F}(B_t,t)| \leq \frac{C_1}{\sqrt t}$ (by construction of $t_0$).
\end{lemma}
 
Synthesizing \Cref{lem:vv_diff} and \Cref{lem:tildehat_diff}, we have $|V(B_t, t) - V_{\hat F}(B_t, t)| \leq (C+C_1)t^{\frac{1}{2}}$. A crucial next step is to relate this bound to the final regret. This is achieved by two transformations. Roughly speaking, the buyer's ``regret'' can be viewed in two parts: 1) she does not bid according to the optimal strategy; 2) her strategy is not optimally spending the budget which leads to future losses. These two transformations are done with this intuitive observation and summarized in the following lemma that bounds the performance of the buyer's strategy. Below we first condition on the good event that the estimation succeeds for every $t$. Then we add the contribution of the bad event to the regret in \Cref{thm:regret_full_feedback}.

\begin{lemma}\label{lem:final_v}
    For any given $B_t$ and $t$, denote $V^\pi (B_t,t)=\E_{\bm v \sim G^T} \left[\sum_{\tau=t}^T \lambda^{\tau-t} R^\pi (v_{\tau},b_{\tau}) \right]$, then
    \[ |V(B_t,t)-V^\pi(B_t,t)|\le \frac{4\left(\sqrt{\frac{1}{2} \ln \frac{2T}{\delta}} \frac{1 + \lambda}{(1 - \lambda)^2}+C_1\right)}{(1-\lambda)\sqrt{t}}. \]
\end{lemma}

By \Cref{lem:final_v} and further transformations, we can now establish the regret bound of \Cref{algo:full_feedback}.
\begin{theorem}\label{thm:regret_full_feedback}
    Under the circumstance that $F$ is unknown, the worse-case regret of \Cref{algo:full_feedback} is $\tilde O(\sqrt T)$, where the regret is computed according to \Cref{eqn:def_regret}. Explicitly, if we take $C_1=1$,
    \[ \text{Regret} \leq \left(4\sqrt{\frac{1}{2} \ln 2T^2} \frac{1 + \lambda}{(1 - \lambda)^2}+5-\lambda \right) \sqrt T + \frac{ \lambda}{2} \log_\frac{1}{\lambda} \frac{T}{(1 - \lambda)^2}+1. \]
\end{theorem}

Note that the regret bound is meaningful only if $B = \omega(T^{\frac{1}{2}})$. Let us now take the budget $B$ to scale linearly with $T$ as in \cite{balseiro2019learning, flajolet2017real}. Specifically, assume that $T<+\infty$ and there exists a constant $\beta$ such that the budget $B = \beta T$, then we establish that the regret is $\tilde O(\sqrt T)$ in this special case. Indeed, under this condition, we can simply set $t_0=T+1$ and $V_{\hat F}(B_{T+1},T+1)=0$ for any $B_{T+1}$ in \Cref{algo:full_feedback}. Therefore, $C_1=0$ and the worst-case regret is bounded by $\left(4\sqrt{\frac{1}{2} \ln 2T^2} \frac{1 + \lambda}{(1 - \lambda)^2}+1-\lambda \right) \sqrt T + \frac{ \lambda}{2} \log_\frac{1}{\lambda} \frac{T}{(1 - \lambda)^2}+1$.

Next we deal with the case where $G$ is also initially unknown. Based on \Cref{algo:full_feedback}, we additionally maintain an estimation $\hat G$ of $G$ based on past observations of valuations. $\hat G$ is initialized to be a uniform distribution and will be used to solve the dynamic programming problem (see \Cref{algoline:estimate_G} of \Cref{algo:censored_feedback}). Using similar techniques as before (with more work), we obtain the following theorem.
\begin{theorem}\label{cor:unknownG}
    Under the circumstance that $F, G$ are both unknown, it holds that the worst-case regret of \Cref{algo:full_feedback} using empirical distribution functions to estimate $F$ and $G$ is $\tilde{O}(\sqrt{T})$. Explicitly, if we take $C_1=1$,
    \[ \text{Regret} \leq \left(\sqrt{\frac{1}{2} \ln 4T^2} \frac{6(1 + \lambda)}{(1 - \lambda)^2}+5-\lambda \right) \sqrt{T}+\frac{\lambda}{2} \log_\frac{1}{\lambda} \frac{T}{(1 - \lambda)^2}+1. \]
\end{theorem}

\subsection{Censored Feedback}\label{subsect:censored}
In this subsection, we deal with the case that the buyer can only see the winner's bid after each round. In other words, the feedback is left-censored. Concretely, the buyer's observation is
\[ o_t = \max\{b_t, m_t\}. \]
When she wins, the exact value of $m_t$ is not revealed. The buyer only knows that $m_t$ is smaller than her bid in the current round. To estimate the distribution of $m_t$, there is a classical statistics (KM estimator) developed by \citet{kaplan1958nonparametric} for the estimation of $F$ in this scenario. However, the KM estimator assumes the sequence $(m_t)_{t=1}^T$ is deterministic, which does not fit our needs. Although \citet{suzukawa2004unbiased}'s extension allows random censorship, it requires independence between $b_t$ and $m_t$, which is not realistic since we use the estimated distribution to place bids.

To tackle this problem, we first introduce an estimator proposed by \citet{zeng2004estimating} denoted by $\hat F_n$ to take place of the previous empirical distribution used in \Cref{algo:full_feedback}.

\paragraph{Estimation procedure}
We now present the procedure for estimating $F$ under censored feedback. It suffices to estimate the distribution function of $1-m_t$ which is right-censored by $1-b_t$. Let $y_t = \min\{1-m_t, 1-b_t\}, r_t = \mathbf 1_{\{m_t \geq b_t\}}$. The observations can now be described as $(y_t, r_t, \bm h_t)_{t=1}^T$. 

Roughly speaking, to decouple the dependency between $m_t, b_t$, we use the fact that $b_t$ and $m_t$ are independent conditioning on $\bm h_t$. Intuitively, the history $\bm h_t$ provides information for getting enough effective samples for $m_t$. Next we establish models to estimate the hazard rate functions\footnote{The hazard rate function of a random variable $X$ with p.d.f. $f$ and c.d.f. $F$ is $H_X(x) = \frac{f(x)}{1-F(x)}$.} of $1-m_t, 1-b_t$ using $\bm h_t$. With the hazard rate functions, we use the maximum likelihood method with a kernel to compute the final estimation $\hat F_t$ and obtain \Cref{eqn:zeng}.

Details follow. We initialize with a sequence (bandwidth) $(a_t)_{t=1}^T$ such that $\frac{\log^2a_t}{ta_t^2} \to 0, ta_t^2 \to \infty, ta_t^4 \to 0$ as $t \to +\infty$ and a symmetric kernel function $K(\cdot, \cdot) \in C^2(\mathbb R^2)$ with bounded gradient. Now, at each time $t$, we compute two vectors $\bm \beta_t, \bm \gamma_t$ which maximize each of the followings
\begin{equation}\label{eqn:cox_likelihood}
    \begin{split}
        \mathcal L_1(\bm \beta) &= \frac{1}{t} \sum_{\tau=1}^t r_\tau \left[\bm \beta^\top \bm h_\tau - \log \left(\sum_{y_i \geq y_\tau} e^{\bm \beta^\top \bm h_\tau}\right) \right], \\
        \mathcal L_2(\bm \gamma) &= \frac{1}{t} \sum_{\tau=1}^t (1-r_\tau) \left[\bm \gamma^\top \bm h_\tau - \log \left(\sum_{y_i \geq y_\tau} e^{\bm \gamma^\top \bm h_\tau}\right) \right].
    \end{split}
\end{equation}
We arbitrarily pad $\bm h_1, \dotsc, \bm h_t$ with zeros to make their length the same (we will show that this is without loss of generality in a moment). Compute $\bm Z_t=(\bm \beta_t^\top \bm h_t, \bm \gamma_t^\top \bm h_t)^\top$. The survival function of $1-m_t$, or equivalently the cumulative distribution function of $m_t$, is estimated based on \citet{zeng2004estimating}'s estimator
\begin{equation}\label{eqn:zeng}
    \hat F_t(x)=\frac{1}{t} \sum_{i=1}^t \prod_{j=1}^t \left(1 - \frac{K((\bm Z_i - \bm Z_j)/a_n) \mathbf 1_{\{y_j \leq x\}} r_j}{\sum_{m=1}^n K((\bm Z_i - \bm Z_m)/a_n) \mathbf 1_{\{y_j \leq y_m\}}} \right).
\end{equation}

Now, we are ready to apply the estimator and the algorithm for censored feedback is depicted in \Cref{algo:censored_feedback}. Note that the new estimator's convergence rate is slower than that for the full-feedback case. Therefore, compared to \Cref{algo:full_feedback}, \Cref{algo:censored_feedback} is now a multi-phase algorithm. The algorithm only updates the estimation of $\hat F$ at the end of each phase (see \Cref{fig:scheme} for an illustration). The other elements of each phase in \Cref{algo:censored_feedback} are similar to \Cref{algo:full_feedback}.

\begin{algorithm}[htbp]
    \caption{Algorithm for the censored-feedback case}
    \label{algo:censored_feedback}
    \begin{algorithmic}[1]
        \State \textbf{Input}: Initial budget $B$ and constant $C_1$ \Comment{$C_1$ is an arbitrary positive constant}
        \State Initialize the estimation $\hat{F}$ of $F$ and the estimation $\hat G$ of $G$ to uniform distributions over $[0, 1]$
        \State $B_1 \gets B$
        \For{Phase $i=1, 2, \dotsc$} \Comment{Phase $i\,(i>1)$ lasts for $2^i$ rounds. Phase 1 lasts for 2 rounds}
            \For{each $t$ in the time interval of round $i$}
                \State Observe the value $v_t$ in round $t$
                \State Update $\hat{G}(x) = \frac{1}{t} \sum_{i=1}^t \mathbf 1_{\{v_i \leq x\}}$. \label{algoline:estimate_G}
                \State Let $t_0$ be the smallest integer that satisfies $\lambda^{t_0-t} \frac{1}{1 - \lambda} < \frac{C_1}{\sqrt t}$
                \State Set $V_{\hat F, \hat G}(B_{t_0}, t_0)=0$ for any $B_{t_0}$ \Comment{$V_{\hat F, \hat G}$ is algorithm's estimation of $V$}
                \For{$\tau=t_0, t_0-1, \dotsc, t$} \Comment{This loop can be moved to the end of each phase to reduce the invocation time from $T$ to $\ln T$}
                    \State $Q_{\hat{F}, \hat G}(B_{\tau}, \tau, b) \gets [(v-b) + \lambda V_{\hat F, \hat G}(B_{\tau}-b, \tau+1)] \hat{F}(b) + \lambda V_{\hat F, \hat G}(B_{\tau}, \tau+1)(1 - \hat{F}(b))$
                    \State Solve the optimization problem $\hat b_{\tau}^* \gets \argmax_b Q_{\hat{F}, \hat G}(B_{\tau}, \tau, b)$
                    \State $V_{\hat F, \hat G}(B_{\tau}, \tau) \gets \E_{v \sim G} [Q_{\hat{F}, \hat G}(B_{\tau}, \tau, \hat b_{\tau}^*)]$
                \EndFor
                \State Place a bid $\hat b_t \gets \argmax_b Q_{\hat{F}, \hat G}(B_t, t, b)$
                \State Observe $o_t, c_t$ and $r_t$ from this round of auction
                \State $B_{t+1} \gets B_t-c_t$. If $B_{t+1} \leq 0$ then halt.
            \EndFor
            \State Update $\hat{F}$ using the estimator specified in \Cref{eqn:zeng} with data observed before this phase
        \EndFor
    \end{algorithmic}
\end{algorithm}

\begin{figure}[htbp]
    \centering
    \input{scheme.tex}
    \caption{Schematic representation of the phases in \Cref{algo:censored_feedback}. \Cref{algo:censored_feedback} updates its estimates of $F$ at the end of each phase.}
    \label{fig:scheme}
\end{figure}
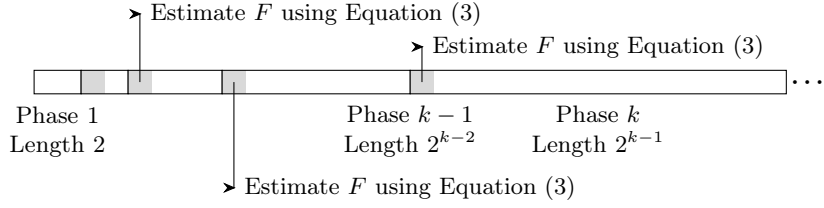

\paragraph{Analysis of regret}
To analyze the performance of \Cref{algo:censored_feedback}, we will prove a series of lemmas on the estimation error of \Cref{eqn:zeng}. In particular, our proof relies on the following convergence result.
\begin{lemma}[Zeng]\label{lem:zeng}
    Let $\hat F_n$ be the estimation of $F$ after using $n$ observations. We have
    \[ \sqrt{n}(\hat{F}_n(1-x)-F(1-x)) \Longrightarrow \mathcal{B}(x) \quad \text{in} \quad \ell^\infty([0,1]), \]
    where $\mathcal{B}(x)$ is a Gaussian process.
\end{lemma}
Before we proceed to apply the lemma, we verify a series of prerequisites mentioned in \cite{zeng2004estimating} to make sure it holds. First, we make sure that conditioning on $\bm h_t$, the random variables $1-m_t$ and $1-b_t$ are independent. Indeed, $b_t$ is completely decided by $\bm h_t$ and $m_t$ is independent of $\bm h_t$. Second, we note that the maximizer shown in \Cref{eqn:cox_likelihood} is essentially doing Cox's proportional hazards regression analysis. To establish consistency of the estimator, we show that at least one of $\tilde m_t \coloneqq 1-m_t$ and $1-b_t$ follows Cox's proportional model. That is to say, there exists $\bm \beta$ and a function $f(y)$ such that the hazard rate function of $\tilde m_t$ or $1-b_t$ conditioning on $\bm h_t$ exactly follows
\begin{equation}\label{eqn:cox}
    H(y \mid \bm h_t)=f(y)e^{\bm \beta^\top \bm h_t}.
\end{equation}
\Cref{eqn:cox} holds for $\tilde m_t$. In fact, taking $\bm \beta = \mathbf 0$ and $f(y) = \frac{F'(1-y)}{1-F(1-y)}$ suffices. Since we take $\bm \beta = \mathbf 0$, consistency establishes regardless of the way we pad $\bm h_t$.

Next, consider some phase at time $2^n\le t \le 2^{n+1}-1$. The estimation $\hat F_n$ is computed using the first $2^n$ observed data points. Applying similar techniques for the rate of convergence of empirical process \citep{norvaivsa1991rate}, we obtain
\begin{lemma}\label{lem:convergencerate}
    Under \Cref{algo:censored_feedback}, let $\hat F_n$ be an empirical process of updates and $\mathcal{B}$ be a general Gaussian process, respectively, indexed by a class $\mathscr{F}$ of real measurable functions. We have:
    \[ |{\Pr(\{\|F_n\|_{\mathscr{F}}\ge r\})-\Pr(\{\|\mathcal{B}\|_{\mathscr{F}}\ge r\})}|\le C_2(1+r)^{-3}\ln^2t \cdot t^{-\frac{1}{6}}, \]
    where $C_2$ is a constant depending only on $\hat F_n$ and $t$ is the size of data used to update the estimation.
\end{lemma}

Synthesizing the convergence results in \Cref{lem:zeng} and \Cref{lem:convergencerate}, we establish the following lemma on the performance of $\hat{F}$ in \Cref{algo:censored_feedback}.

\begin{lemma}\label{lem:TSN}
    Under the update process in \Cref{algo:censored_feedback}, for any $2^n\le t\le 2^{n+1}-1$, we have the following bounds for the estimation $\hat F_n$:
    \[ |{\Pr(\sup_b|\sqrt{2^n}(\hat{F}_n(1-b)-F(1-b))|\ge r)-\Pr(\sup_b|\mathcal{B}(1-b)|\ge r)}|\le M (1+r)^{-3}\ln^2 t\cdot t^{-\frac{1}{6}}, \]
    where $M$ is a constant depending only on $F$ and \Cref{algo:censored_feedback}.
\end{lemma}

Finally, we now bound the difference between $\hat{F}_n$ and $F$ with the help of \Cref{lem:TSN}.
\begin{lemma}\label{lem:Fbound}
    Recall that we use the first $2^n$ data points to estimate $\hat F$. Under the update procedure of \Cref{algo:censored_feedback}, for any $2^n\le t\le 2^{n+1}-1$, with probability at least $1-T^{-\frac{5}{12}}/(2 \ln T)$
    \[\sup_x|\hat F(x)-F(x)|\le \sqrt{2}C_5(4M\ln^3T)^\frac{1}{3}t^{-\frac{5}{9}}T^{\frac{5}{36}}.
    \]
\end{lemma}

With \Cref{lem:Fbound} in hand and employing a similar methodology in the proof of \Cref{thm:regret_full_feedback}, we now have
\begin{theorem}\label{thm:censored_regret}
    Under the circumstance that $F, G$ are both unknown and the feedback is censored, the worst-case regret of \Cref{algo:censored_feedback} is $\tilde{O}(T^{\frac{7}{12}})$. Explicitly, if we take $C_1=1$,
    \begin{gather*}
        \text{Regret} \leq \left(\frac{9 \sqrt 2(1 + \lambda)}{2(1 - \lambda)^2}C_5(4M \ln^3T)^{\frac{1}{3}}+1\right)T^{\frac{7}{12}}\\+\left(\sqrt{\frac{1}{2}\ln{4T^{\frac{17}{12}}}}\frac{2(1+\lambda)}{(1-\lambda)^2}+5 -\lambda\right) \sqrt{T}+\frac{ \lambda}{2} \log_\frac{1}{\lambda} \frac{T}{(1 - \lambda)^2}.
    \end{gather*}
\end{theorem}

\begin{remark}
    The previous setting in \cite{han2020optimal} is a special case of our setting. In \citet{han2020optimal}, there are no budget constraints and $\lambda=0$ (thus removing the $\frac{1}{1 - \lambda}$ factor in our results). The buyer's aim is to maximize $(v-b)F(b)$ each round. This is equivalent to the case $V_{\hat F}=0$ in our setting with no need to estimate $G$, yielding regret $\tilde O(\sqrt T)$ in the full-feedback case and regret $\tilde O(T^{\frac{7}{12}})$ in the censored-feedback case. 
\end{remark}

\section{Discussion and Conclusion}
In this paper, we develop a learning algorithm to adaptively bid in repeated first-price auctions with budgets. On the theoretical side, our algorithm, together with its analysis of $\tilde O(\sqrt T)$-regret in the full-feedback case and $\tilde O(T^{\frac{7}{12}})$-regret in the censored-feedback case, takes the first step in understanding the problem. On the practical side, our algorithm is simple and readily applicable to the digital world that has shifted to first-price auctions.

Questions raise themselves for future explorations. We observe here that in the limiting case $\lambda \to 1$, the optimal bidding strategy in \Cref{algo:censored_feedback} is similar to a \emph{pacing} strategy, which partially answers the open question\footnote{``A question of theoretical and practical interest is how to extend the adaptive pacing approach that is suggested in this paper to nontruthful mechanisms such as first-price auctions.''} raised in \cite{balseiro2019learning}. In the limiting case of $\lambda \to 1$, the optimal bid of \Cref{algo:censored_feedback} is of the form $\frac{v_t}{1+x_t}$, where $x_t$ is a pacing multiplier that depends only on $B_t$ and $F$. Further, $x_t$ can be computed without explicitly solving the dynamic programming problem. This observation can be viewed as a direct corollary of \cite[Theorem 3.1][]{gummadi2012repeated}. This connection between \Cref{algo:censored_feedback} and pacing suggests future directions to extend the current adaptive pacing techniques to other auction forms\footnote{We note here that parallel to our work, \citet{2022arXiv220508674G} extend the pacing techniques to a class of auction forms (called core auctions). They obtain $\tilde O(T^{\frac{3}{4}})$-regret bounds against \emph{the best linear policy} under the \emph{value-maximization} objective.}. Other immediate open questions include establishing lower bounds of regret for our problem.

{\small\bibliographystyle{plainnat}
\bibliography{references.bib}}

\newpage
\appendix

\input{appendix.tex}

\end{document}

%% file: scheme.tex
\begin{tikzpicture}[scale=1.25]
    \footnotesize
    \draw (-1,0) rectangle ++(0.5,0.25) {};
    \draw[fill=gray!30, draw=none] (-0.5,0) rectangle ++(0.25,0.25) {}; \draw (-0.5,0) rectangle ++(0.5,0.25) {};
    \draw[fill=gray!30, draw=none] (0,0) rectangle ++(0.25,0.25) {}; \draw (0,0) rectangle ++(1,0.25) {};
    \draw[fill=gray!30, draw=none] (1,0) rectangle ++(0.25,0.25) {}; \draw (1,0) rectangle ++(2,0.25) {};
    \draw[fill=gray!30, draw=none] (3,0) rectangle ++(0.25,0.25) {}; \draw (3,0) rectangle ++(4,0.25) {};

    \node (4) at (5,0.5) {Estimate $F$ using \Cref{eqn:zeng}};
    \node (3) at (3,-1) {Estimate $F$ using \Cref{eqn:zeng}};
    \node (2) at (2,0.85) {Estimate $F$ using \Cref{eqn:zeng}};

    \draw[-{Stealth[angle'=45]}] (3.125,0.125) |- (4);
    \draw[-{Stealth[angle'=45]}] (1.125,0.125) |- (3);
    \draw[-{Stealth[angle'=45]}] (0.125,0.125) |- (2);

    \node at (5,-0.4) {\begin{tabular}{c} Phase $k$ \\ Length $2^{k-1}$ \end{tabular}};
    \node at (3,-0.4) {\begin{tabular}{c} Phase $k-1$ \\ Length $2^{k-2}$ \end{tabular}};
    \node at (-0.75,-0.4) {\begin{tabular}{c} Phase 1 \\ Length 2 \end{tabular}};

    \node at (7.25,0.125) {$\bm \dotsb$};
\end{tikzpicture}

%% file: appendix.tex
\section{Omitted Proofs in \Cref{subsec:full_feedback}}
\subsection{Proof of \Cref{lem:vv_diff}}
We will use backward induction to show that
\[| V(B_t, t) - V_{\hat F}(B_t, t)| \leq Ct^{-\frac{1}{2}}. \]

The inequality holds trivially with the condition of the lemma for the basis ($t=t_0$). Suppose the bound holds for $t+1$. Now we write out the difference of the value functions
\begin{align*}
    V(B_t, t) - V_{\hat F}(B_t, t) &= \E_{ v \sim G} [Q_{v, F}(B_t, t, b^*_t)-Q_{v, \hat{F}}(B_t, t, b_t)]
    \\& \leq \E_{ v \sim G} [Q_{v, F}(B_t, t, b^*_t)-Q_{v, \hat{F}}(B_t, t, b_t^*)],
\end{align*}
where $b_t$ is \Cref{algo:full_feedback}'s estimated optimal bid and $b^*_t$ is the bid of the benchmark. The inequality establishes by noting that $b^*_t$ is sub-optimal under $\hat{F}$. Next consider the term inside the expectation which is rewritten as follows:
\begin{align*}
    &Q_{v, F}(B_t, t, b^*_t)-Q_{v, \hat{F}}(B_t, t, b^*_t) \leq \notag
    \\& \underbrace{|(v_t-b^*_t)(F(b^*_t) - \hat{F}(b^*_t))|}_{\Delta_1}
    \\&+ \underbrace{\lambda F(b^*_t)|(V(B_t-b^*_t, t+1) - V_{\hat F}(B_t-b^*_t, t+1))| + \lambda|(F(b^*_t) - \hat{F}(b^*_t)) V_{\hat F}(B_t-b^*_t, t+1)|}_{\Delta_2}
    \\&+ \underbrace{\lambda|(1 - F(b^*_t))(V(B_t, t+1) - V_{\hat F}(B_t, t+1))| + \lambda|(\hat{F}(b^*_t) - F(b^*_t)) V_{\hat F}(B_t, t+1)|}_{\Delta_3}.
\end{align*}
To bound the above equation, we deal with the three terms $\Delta_1, \Delta_2, \Delta_3$ separately. Using \Cref{lem:DKW} and union bound of $T$ rounds, $\Delta_1 \leq \sqrt{\frac{1}{2} \ln \frac{2T}{\delta}} t^{-\frac{1}{2}}$ with probability at least $1 - \delta$. By the induction hypothesis $|V(B_t, t) - V_{\hat F}(B_t, t)| \leq C(t+1)^{-\frac{1}{2}} \leq Ct^{-\frac{1}{2}}$ and that any value function is bounded by $1 + \lambda + \lambda^2 + \dotsb = \frac{1}{1 - \lambda}$, we have
\begin{align*}
    \Delta_2 &\leq \lambda CF(b^*_t)t^{-\frac{1}{2}} + \frac{\lambda}{1 - \lambda} \sqrt{\frac{1}{2} \ln \frac{2T}{\delta}} t^{-\frac{1}{2}},
    \\ \Delta_3 &\leq \lambda C(1-F(b^*_t))t^{-\frac{1}{2}} + \frac{\lambda}{1 - \lambda} \sqrt{\frac{1}{2} \ln \frac{2T}{\delta}} t^{-\frac{1}{2}}.
\end{align*}

Therefore,
\begin{align*}
    Q_{v, F}(B_t, t, b^*_t)-Q_{v, \hat{F}}(B_t, t, b^*_t) &\leq \left(\frac{2 \lambda}{1 - \lambda} \sqrt{\frac{1}{2} \ln \frac{2T}{\delta}} + \lambda C + \sqrt{\frac{1}{2} \ln \frac{2T}{\delta}} \right) t^{-\frac{1}{2}} \\&= \sqrt{\frac{1}{2} \ln \frac{2T}{\delta}} \frac{1 + \lambda}{(1 - \lambda)^2} t^{-\frac{1}{2}}.
\end{align*}
This establishes that $V(B_t, t) - V_{\hat F}(B_t, t)\le Ct^{-\frac{1}{2}}$. Finally, by symmetry---swap $F$ and $\hat F$ and repeat the proof above, it holds that
\[ |V(B_t, t) - V_{\hat F}(B_t, t)|\le Ct^{-\frac{1}{2}}. \]
This finishes the induction step and the claim follows. 

\subsection{Proof of \Cref{lem:tildehat_diff}}
We will state a general form of \Cref{lem:tildehat_diff} concerning the error in the initialization of the base case. This lemma will come in handy in the following sections.
\begin{lemma}\label{lem:diff_vvforanydist}
    For any fixed distributions $F, G$, consider the value function $V_{F, G}(B_t, t)$. Suppose we use an arbitrary value in $\left[0, \frac{1}{1-\lambda}\right]$ to initialize the base case $V_{F, G}(B_{t_0}, t_0)$ and recursively compute thereon to obtain $\tilde V_{F, G}(B_t, t)$, then it holds that for any $t \le t_0$:
    \[ \sup_{B_t} |V_{F, G}(B_t,t)-\tilde{V}_{F, G}(B_t,t)|\le \frac{1}{1-\lambda}\lambda^{t_0-t}. \]
\end{lemma}
\begin{proof}
When $\tau=t_0$, the claim holds because $V_{F, G}$ and $\tilde{V}_{F, G}(\cdot,\cdot)$ are both upper bounded by $\frac{1}{1-\lambda}$ and lower bounded by 0.

Supposing the claim holds when $\tau=t+1$, then for $\tau=t$, we have
\begin{align*}
    \tilde{V}_{F, G}(B_{t},t)&-V_{F, G}(B_{t},t) \le \E_{ v\sim G}[(v_{t}-b_{t}^*){F}(b_{t}^*)+\lambda{F}(b_{t}^*)\tilde{V}_{F, G}(B_{t}-b_t^*,t+1)
    \\&+\lambda(1-{F}(b_{t}^*))\tilde{V}(B_{t},t+1)-(v_{t}-b_{t}^*){F}(b_{t}^*)-\lambda{F}(b_{t}^*)V_{F, G}(B_{t}-b_t^*,t+1)
    \\&-\lambda(1-{F}(b_{t}^*))V_{F, G}(B_{t},t+1)]
    \\&\le \frac{1}{1-\lambda}\lambda^{t_0-t-1}\lambda
    \\&=\frac{1}{1-\lambda}\lambda^{t_0-t},
\end{align*}
In the derivation above, $b_t^*$ denotes the optimal bidding strategy obtained by computing $V(B_t,t)$. The first inequality holds since $b^*_{t}$ is not be the optimal bidding strategy under $\hat {V}(\cdot,t)$'s view. The second inequality holds since $|\tilde V(B_{t+1},t+1)- V(B_{t+1},t+1)|\le \frac{1}{1-\lambda}\lambda^{t_0-t-1}$ for any $B_{t+1}$. 

And by symmetry, we have $|\tilde V(B_{t},t)- V(B_{t},t)|\le \frac{1}{1-\lambda}\lambda^{t_0-t}$. This concludes the induction step and yields the lemma.
\end{proof}
In particular, using \Cref{lem:diff_vvforanydist}, under the condition of \Cref{lem:tildehat_diff}, the initialization is taken to be 0. We have
\[ \sup_{B_t}|V_{\hat F}(B_t,t) - \tilde{V}_{\hat F}(B_t,t)|\le  \frac{1}{1-\lambda}\lambda^{t_0-t}. \]
\begin{remark}
    For convenience, similar to the notations used in this lemma, for any value function $\nu$, we will use $\tilde \nu$ to denote its approximately-initialized counterpart. Furthermore, we will invoke the lemma many times for other value functions in the rest of the proofs.
\end{remark}

\subsection{Proof of \Cref{thm:regret_full_feedback}}
Below we first condition on the good event that the estimation (of \Cref{lem:DKW}) succeeds for every $t$. Then we add the contribution of the bad event to the regret in \Cref{thm:regret_full_feedback} finally.

To proof this theorem, we will first bound the following auxiliary ``regret'' with \Cref{lem:vv_diff} and \Cref{lem:diff_vvforanydist}.

Let us first make an intuitive and approximate description of the regret. The buyer's ``regret'' can be viewed in two parts: 1) she does not bid according to the optimal strategy; 2) her strategy is not optimally spending the budget which leads to future losses. Given remaining budget $B_t$ at time $t$ with strategy $\pi$, the above intuition guides us to first look at
\begin{align*}
    R_1 &= \E_{\bm v \sim G^T} \Biggl[\sum_{t=1}^T (R^{\pi^*}_t(v_t,b^*_t)-R^{\pi}_t(v_t,b_t)) \\&+ [\lambda(F(b_t^*) V(B_t-b_t^*, t+1)+(1-F(b_t^*)) V(B_t, t+1)) \\&- \lambda(F(b_t) V(B_t-b_t, t+1)+(1-F(b_t)) V(B_t, t+1))] \Biggr].
\end{align*}

\begin{lemma}\label{lem:r1}
    Suppose \Cref{lem:vv_diff} and \Cref{lem:diff_vvforanydist} hold for some constants $C$ and $C_1$, i.e. $|V(B_t, t)-V_{\hat F}(B_t, t)| \leq (C+C_1)t^{-\frac{1}{2}}$. Assume further that $\sup_x |F(x) - \hat F(x)| \leq Kt^{-\frac{1}{2}}$ for some constant $K$. We have
    \[ R_1 \leq 2\left(\frac{K(1+\lambda)}{1-\lambda}+(1+\lambda)C+2C_1\right) \sqrt T. \]
\end{lemma}
\begin{proof}
To ease description, we first let
\begin{align*}
    \hat H_t &\coloneqq (v_t-b_t)\hat{F}(b_t) + \lambda(\hat{F}(b_t) V_{\hat F}(B_t-b_t, t+1)+(1-\hat{F}(b_t))V_{\hat F}(B_t, t+1)), \\
    H_t &\coloneqq (v_t-b_t^*)F(b_t^*) + \lambda(F(b_t^*)V(B_t-b_t^*, t+1)+(1-F(b_t^*)) V(B_t, t+1)), \\
    \tilde H_t &\coloneqq (v_t-b_t)F(b_t) + \lambda(F(b_t)V(B_t-b_t, t+1)+(1-F(b_t))V(B_t, t+1)).
\end{align*}
(Recall that $b_t$ is \Cref{algo:full_feedback}'s estimated optimal bid and $b^*_t$ is the bid of the benchmark.) Using the notations above, $R_1$ now becomes
\[ \E_{\bm v \sim G^T} \left[\sum_{t=1}^T (H_t - \tilde H_t) \right] \leq \E_{\bm v \sim G^T} \left[\sum_{t=1}^T(|H_t - \hat H_t|+|\hat H_t - \tilde H_t|) \right]. \]
Use the induction part in the proof of \Cref{lem:vv_diff} and \Cref{lem:diff_vvforanydist}, $|H_t - \hat H_t| \leq (C+C_1)t^{-\frac{1}{2}}$ follows from the condition. In order to bound $|\hat H_t - \tilde H_t|$, we write
\begin{align*}
    |\hat H_t - \tilde H_t| \leq &|(v_t-b_t)(F(b_t) - \hat F(b_t))| + \lambda|\hat F(b_t)-F(b_t)| V_{\hat F}(B_t-b_t, t+1)
    \\&+ \lambda F(b_t)|V(B_t-b_t, t+1) - V_{\hat F}(B_t-b_t, t+1)|
    \\&+ \lambda|F(b_t) - \hat F(b_t)| V_{\hat F}(B_t, t+1) + \lambda(1-F(b_t))|V(B_t, t+1) - V_{\hat F}(B_t, t+1)|
    \\& \leq Kt^{-\frac{1}{2}} \left(1 + \frac{2 \lambda}{1 - \lambda} \right) + \lambda \left(Ct^{-\frac{1}{2}} + \frac{C_1}{\lambda} {t}^{-\frac{1}{2}} \right)
    \\&= \left(\frac{K(1+\lambda)}{1-\lambda}+\lambda C+C_1\right)t^{-\frac{1}{2}}.
\end{align*}
The first inequality holds because of the triangle inequality and the second inequality establishes using the conditions specified (\Cref{lem:DKW}, \Cref{lem:vv_diff} and \Cref{lem:diff_vvforanydist}). Note that it holds that $|V(B_{t+1},t+1)-V_{\hat F}(B_{t+1},t+1)|\le \frac{1}{1-\lambda}\lambda^{t_0-t-1}=\frac{C_1}{\lambda\sqrt{t}}.$

Finally, we sum up the above estimation and obtain
\[ R_1 \leq \sum_{t=1}^T \left(\frac{K(1+\lambda)}{1-\lambda}+(1+\lambda)C+2C_1\right) t^{-\frac{1}{2}}, \]
as desired.
\end{proof}
In particular, taking $C = \sqrt{\frac{1}{2} \ln \frac{2T}{\delta}} \frac{1 + \lambda}{(1 - \lambda)^2}$ and $K = \sqrt{\frac{1}{2} \ln \frac{2T}{\delta}}$ as it holds in the analysis of \Cref{algo:full_feedback}, we have $R_1 \leq 4(C+C_1) \sqrt T$.

Next, we will build connections between $R_1$ and the accumulated differences in the values. The definition of $R_2$ is inspired by the recent advances in \cite{yang2021q, he2021nearly, liu2020regret, zhou2021provably}. We define
\[ R_2=\E_{\bm v \sim G^T} \left[\sum_{t=1}^T (V(B_t,t)-V^\pi (B_t,t)) \right]. \]

\begin{lemma}\label{lem:r2}
    For any given $B_t$ and $t$, suppose that the conditions for \Cref{lem:r1} holds. We have
    \[ |V(B_t,t)-V^\pi(B_t,t)|\le \frac{C'}{(1-\lambda)\sqrt{t}} \quad \text{where} \quad C'=\frac{K(1+\lambda)}{1-\lambda}+(1+\lambda)C+2C_1. \]
\end{lemma}

\begin{proof}
To start, we introduce the following notation:
\[ H_t^\pi=(v_t-b_t)F(b_t) + \lambda(F(b_t)V^\pi(B_t-b_t, t+1)+(1-F(b_t))V^\pi(B_t, t+1)). \]
Note that we have $\E_{\bm v \sim G^T}[H_t]=V(B_t, t)$ and $\E_{\bm v\sim G^T}[H_t^\pi]=V^\pi(B_t, t)$. We will use backward induction to bound the difference between $V(B_t,t)$ and $V^\pi(B_t,t)$ for any $B_t$. First choose a sufficient large $t_0$ and assume $V(B_{t_0},t_0)=\tilde V^\pi(B_{t_0},t_0)$. When $\tau=t_0$, the induction basis holds since $V(B_{t_0},t_0)=\tilde V^\pi(B_{t_0},t_0)$. Now consider $\tau=t+1$. By the induction hypothesis, $|V(B_\tau,\tau)-\tilde V^\pi(B_\tau,\tau)|\le \frac{C'}{(1-\lambda)\sqrt{\tau}}$, and when $\tau=t$, it holds that
\begin{align*}
    \tilde H_t-H_t^\pi &=\E_{v\sim G}[\lambda F(b_t)(V(B_t-b_t,t+1)-\tilde V^\pi(B_t-b_t,t+1))\\
    &+\lambda (1-F(b_t))(V(B_t,t+1)-\tilde V^\pi(B_t,t+1))]\\
    &\le \frac{C'\lambda}{(1-\lambda)\sqrt{t+1}}.    
\end{align*}
 
It follows that
\[ |V(B_t,t)-\tilde V^\pi(B_t,t)|\le \E_{\bm v\sim G}[|H_t-\tilde H_t|+|\tilde H_t-H_t^\pi|] \le \frac{C'}{\sqrt{t}}+\frac{C'\lambda}{(1-\lambda)\sqrt{t+1}}\le \frac{C'}{(1-\lambda)\sqrt{t}}. \]
and this concludes the induction step.
 
Applying the proof techniques in \Cref{lem:diff_vvforanydist} with the condition for $V^\pi$, we have 
\[ |V^\pi(B_t,t)-\tilde V^\pi(B_t,t)|\le \frac{1}{1-\lambda}\lambda^{t_0-t}. \]
Since $t_0$ is arbitrarily chosen (as long as it is sufficiently large), we can take $t_0 \to +\infty$. Note that we have $\lim_{t_0\to +\infty} \tilde V^\pi(B_t,t)=V^\pi(B_t,t)$. Therefore, it holds that
\[ |V(B_t,t)-V^\pi(B_t,t)|\le \frac{C'}{(1-\lambda)\sqrt{t}}. \]
This establishes
\[ R_2\le \frac{2C'}{1-\lambda}\sqrt{T}, \]
and concludes the proof.
\end{proof}
In particular, take $C = \sqrt{\frac{1}{2} \ln \frac{2T}{\delta}} \frac{1 + \lambda}{(1 - \lambda)^2}$ and $K = \sqrt{\frac{1}{2} \ln \frac{2T}{\delta}}$ as it holds in the analysis of \Cref{algo:full_feedback}, we have $R_2 \leq \frac{4(C+C_1)}{1-\lambda} \sqrt T$.

Finally, we relate the regret defined in \Cref{eqn:def_regret} with $R_2$ in the previous lemma.
\begin{lemma}\label{lem:final_regret}
    Suppose that $R_2 \leq \frac{2C'}{1 - \lambda} \sqrt T$, then $\text{Regret} \leq (2C'+1-\lambda) \sqrt T + \frac{1 - \lambda}{2} \log_\frac{1}{\lambda} \frac{T}{1 - \lambda}$.
\end{lemma}
\begin{proof}
    In fact, note that 
    \[ R_2=\E_{\bm v \sim G^T} \left[\sum_{t=1}^T (1+\lambda+ \dotsb +\lambda^{t-1})(R^*(v_t,b_t^*)-R^\pi(v_t,b_t)) \right],
    \]
    and any expected instantaneous reward is no greater than 1. We have
    \begin{gather*}
        \left|\frac{1}{1-\lambda} \text{Regret}-R_2 \right|=\E_{\bm v \sim G^T} \left[\left|\sum_{t=1}^T \frac{\lambda^t}{1-\lambda}(R^*(v_t,b_t^*)-R^\pi(v_t,b_t)) \right| \right] \\\le \frac{\lambda}{2(1-\lambda)}\log_\frac{1}{\lambda}\frac{T}{(1-\lambda)^2}+\frac{1}{\sqrt{T}}T.
    \end{gather*}
    The above bound holds because when $t\ge \log_{\frac{1}{\lambda}}\frac{\sqrt{T}}{1-\lambda}$, we have $\frac{\lambda^t}{1-\lambda}\le \frac{1}{\sqrt{T}}$. The considered quantity is divided into two parts with this threshold $\log_{\frac{1}{\lambda}}\frac{\sqrt{T}}{1-\lambda}$. The first part is no greater than $\log_{\frac{1}{\lambda}}\frac{\sqrt{T}}{1-\lambda}\frac{\lambda}{1-\lambda}$ and the second part is no greater than $\frac{1}{\sqrt{T}}T=\sqrt{T}$. The result follows after simple rearrangements.
\end{proof}
In particular, take $C' = \frac{K(1+\lambda)}{1-\lambda}+(1+\lambda)C+2C_1$ as it holds in the analysis of \Cref{algo:full_feedback}, we have
\[ \text{Regret} \leq \left(\left(4\sqrt{\frac{1}{2} \ln \frac{2T}{\delta}} \frac{1 + \lambda}{(1 - \lambda)^2}+4C_1\right)+1-\lambda \right) \sqrt T + \frac{\lambda}{2} \log_\frac{1}{\lambda} \frac{T}{(1 - \lambda)^2}, \]
conditioning on the good event that the estimation (of \Cref{lem:DKW}) succeeds for every $1 \leq t \leq T$. Now take $\delta = \frac{1}{T}$ and note that $\Pr[\text{bad event}] \leq \frac{1}{T}$. By using the trivial regret bound $T$ for the failure event, we have
\[ \text{Regret} \leq \left(\left(4\sqrt{\frac{1}{2} \ln 2T^2} \frac{1 + \lambda}{(1 - \lambda)^2}+4C_1\right)+1-\lambda \right) \sqrt T + \frac{\lambda}{2} \log_\frac{1}{\lambda} \frac{T}{(1 - \lambda)^2} + \frac{1}{T} \cdot T, \]
and this concludes the proof of \Cref{thm:regret_full_feedback}.

\subsection{Proof of \Cref{cor:unknownG}}
Note that the function $Q$ depends on both $F$ and $G$ as $V$ does. With the additional condition that $G$ also needs to be estimated, we use $Q_{\cdot, \hat G}$ and $Q_{\cdot, G}$ to denote the corresponding version computed using $\hat G$ and $G$ respectively. We also extend this notation on $V$ in this scenario (c.f. \Cref{algo:censored_feedback}).

We start with two simple lemmas. These two lemmas are direct corollaries of \Cref{lem:DKW} and \Cref{lem:diff_vvforanydist}.
\begin{lemma}\label{lem:DKW_G}
    Let $\hat G_t$ be the estimated distribution at round $t$ in \Cref{algo:full_feedback}. With probability $1-\delta$,
    \[ \sup_x|\hat G_t(x)-G(x)|\le \sqrt{\frac{1}{2}\ln\frac{2}{\delta}}t^{-\frac{1}{2}}. \]
\end{lemma}

\begin{lemma}\label{lem:unknownG_diff_tilde_hat}
    For any round $t\le  T$, budget $B_t$, it holds that
    \[ \sup_{B_t}|V_{\hat F,\hat G}(B_t,t)-\tilde V_{\hat F, \hat G}(B_t,t)|\le \frac{1}{1-\lambda}\lambda^{t_0-t}.\]
\end{lemma}

In the following proof, we will bridge $V(B_t,t)$ and $V_{\hat F, \hat G}(B_t, t)$---the estimated value function, gradually. For distributions $A, B$, the notation $V_{A, B}$ refers to the value function computed when $F=A$ and $G=B$.
\begin{lemma}\label{lem:unknownG_diff_vv}
    With probability at least $1-\frac{\delta}{2T}$, for any given $t\le T$ and budget $B_t$, we have
    \[ |V_{\hat F,\hat G}(B_t,t)-V_{\hat F,G}(B_t,t)|\le \sqrt{\frac{1}{2}\ln\frac{4T}{\delta}} \frac{1}{(1-\lambda)^2} t^{-\frac{1}{2}}. \]
\end{lemma}

\begin{proof}
First we note that \Cref{lem:DKW_G} states with probability at least $1-\frac{\delta}{2T}$, $\sup_x|G_t(x)-G(x)|\le \sqrt{\frac{1}{2}\ln\frac{4T}{\delta}}t^{-\frac{1}{2}}$. Now we apply backward induction. When $t=t_0$, the induction basis holds trivially since $|V_{\hat F,\hat G}(B_t,t)-V_{\hat F,G}(B_t,t)|=0\le \sqrt{\frac{1}{2}\ln\frac{4T}{\delta}} \frac{1}{1-\lambda} t^{-\frac{1}{2}}$. Assume the induction hypothesis holds for $\tau=t+1$. For any $t<t_0, v_t$ and $B_t$, it holds that
\begin{align*}
    Q_{\hat F, G}(B_t,t,b_t^*)&-Q_{\hat F, \hat G}(B_t,t,b_t)\le Q_{\hat F, G}(B_t,t,b_t^*)-Q_{\hat F, \hat G}(B_t,t,b_t^*)\\
    &=\lambda \hat{F}(b_t^*)V_{\hat F,G}(B_t-b_t^*,t+1)+\lambda (1-\hat{F}(b_t^*))V_{\hat F, G}(B_t,t+1)\\
    &-\lambda \hat{F}(b_t^*)V_{\hat F, G}(B_t-b_t^*,t+1)-\lambda (1-\hat{F}(b_t^*))V_{\hat F,\hat G}(B_t,t+1)\\
    &\le \lambda \hat{F}(b_t^*)|V_{\hat F, G}(B_t-b_t^*,t+1)-V_{\hat F,\hat  G}(B_t-b_t^*,t+1)|\\
    &+\lambda (1-\hat{F}(b_t^*))|V_{\hat F, G}(B_t,t+1)-V_{\hat F, \hat G}(B_t,t+1)|\\
    &\le \sqrt{\frac{1}{2}\ln\frac{4T}{\delta}} \frac{\lambda}{(1-\lambda)^2} t^{-\frac{1}{2}}.
\end{align*}
The first inequality holds since by construction $b_t^*$ is the optimal bid under $\hat F$ and $G$'s view while $b_t$ is the optimal bid under $\hat F$ and $\hat G$'s view. The second inequality holds with the triangle inequality. And the third inequality is due to the induction hypothesis. Finally, by symmetry---swap $G$ and $\hat G$ and repeat the proof above, it holds that $|Q_{\hat F, G}(B_t,t,b_t)-Q_{\hat F, \hat G}(B_t,t,b_t^*)|\le\sqrt{\frac{1}{2}\ln\frac{4T}{\delta}} \frac{\lambda}{(1-\lambda)^2} t^{-\frac{1}{2}}$.

Next, since $V_{\hat F,\hat G}(B_t,t)=\E_{\hat G}[Q_{\hat F, \hat G}(B_t,t,b_t)]$ and $V_{\hat F,G}(B_t,t)=\E_{ G}[Q_{\hat F, v_t}(B_t,t,b_t^*)]$, we have
\[ |V_{\hat F,G}(B_t,t)-V_{\hat F,\hat G}(B_t,t)|\le \Delta_1+\Delta_2, \]
where
\begin{align*}
    \Delta_1&=|{\E_{ G}[Q_{\hat F, G}(B_t,t,b_t^*)]-\E_{\hat G}[Q_{\hat F, G}(B_t,t,b_t^*)]}| \text{ and} \\
    \Delta_2&=|{\E_{\hat G}[Q_{\hat F, \hat G}(B_t,t,b_t)]-\E_{\hat G}[Q_{\hat F, G}(B_t,t,b_t^*)]}|.
\end{align*}

\begin{itemize}
    \item To bound $\Delta_1$: since $Q_{\hat F, G}$ is supported on $\left[0,\frac{1}{1-\lambda} \right]$ and the difference between $G$ and $\hat{G}$ is upper bounded by $\sqrt{\frac{1}{2}\ln\frac{4T}{\delta}}t^{-\frac{1}{2}}$, $\Delta_1$ is bounded by $\sqrt{\frac{1}{2}\ln\frac{4T}{\delta}}t^{-\frac{1}{2}}\frac{1}{1-\lambda}$. Note that we use the fact that with integration by parts, $\int_0^1 Q\,\mathrm d(G-\hat G)=Q(G-\hat G)|_0^1 - \int_0^1 (G-\hat G)\,\mathrm dQ$. Therefore, it holds that $|\Delta_1|\le\int_0^1 |G-\hat G||\mathrm dQ|$. Since $Q$ is monotone w.r.t. $v_t$ and is two-sided bounded, it holds that $|\Delta_1|\le\sqrt{\frac{1}{2}\ln\frac{4T}{\delta}}t^{-\frac{1}{2}}\frac{1}{1-\lambda}$.
    \item To bound $\Delta_2$, it is clear that $\Delta_2\le \sqrt{\frac{1}{2}\ln\frac{4T}{\delta}} \frac{\lambda}{(1-\lambda)^2} t^{-\frac{1}{2}} $ by linearity of expectation.
\end{itemize}

Therefore, we obtain
\[ |V_{\hat F,G}(B_t,t)-V_{\hat F,\hat G}(B_t,t)|\le \sqrt{\frac{1}{2}\ln\frac{4T}{\delta}} \frac{1}{(1-\lambda)^2} t^{-\frac{1}{2}}, \]
which finishes induction step and concludes the proof.
\end{proof}

\begin{lemma}
For any $t\le T$ and budget $B_t$, with probability at least $1-\frac{\delta}{T}$, it holds that:
\[ |V_{F,G}(B_t,t)-\tilde V_{\hat F,\hat G}(B_t,t)|\le \left(\sqrt{\frac{1}{2} \ln \frac{4T}{\delta}} \frac{2 + \lambda}{(1 - \lambda)^2}+C_1 \right) t^{-\frac{1}{2}}. \]
\end{lemma}

\begin{proof}
In order to bound the difference between $\tilde V_{\hat F,\hat G}(B_t,t)$ and $V_{F,G}(B_t,t)$. We first rewrite
\[ |V_{F,G}(B_t,t)-\tilde V_{\hat F,\hat G}(B_t,t)|\le \Delta_1+\Delta_2+\Delta_3, \]
where
\begin{align*}
    \Delta_1&=|\tilde V_{\hat F,\hat G}(B_t,t)-V_{\hat F,\hat G}(B_t,t)|,\\ \Delta_2&=|V_{\hat F,\hat G}(B_t,t)-V_{\hat F,G}(B_t,t)| \text{ and}\\
    \Delta_3&=|V_{\hat F,G}(B_t,t)-V_{F, G}(B_t,t)|.
\end{align*}

\begin{itemize}
    \item To bound $\Delta_1$: using \Cref{lem:unknownG_diff_tilde_hat} and the definition of $t_0$, we conclude that $\Delta_1\le \frac{C_1}{\sqrt{t}}$.
    \item To bound $\Delta_2$: using \Cref{lem:unknownG_diff_vv}, it holds that $\Delta_2\le \sqrt{\frac{1}{2}\ln\frac{4T}{\delta}} \frac{1}{(1-\lambda)^2} t^{-\frac{1}{2}}$.
    \item To bound $\Delta_3$: using \Cref{lem:vv_diff}, we obtain that $\Delta_3\le\sqrt{\frac{1}{2} \ln \frac{4T}{\delta}} \frac{1 + \lambda}{(1 - \lambda)^2} t^{-\frac{1}{2}}$. Note that we use Bonferroni's method to divide $\delta$ into two parts for the union bound of error in the estimations of $F$ and $G$.
\end{itemize}

Therefore, it holds that:
\[ |V_{F,G}(B_t,t)-\tilde V_{\hat F,\hat G}(B_t,t)|\le \left(\sqrt{\frac{1}{2} \ln \frac{4T}{\delta}} \frac{2 + \lambda}{(1 - \lambda)^2}+C_1 \right) t^{-\frac{1}{2}}. \]
\end{proof}

Now, we provide a bound similar to \Cref{lem:r1} for \Cref{cor:unknownG}. Similarly, we define
\begin{align*}
    \hat{H}_t &\coloneqq (v_t-b_t)\hat{F}(b_t)+\lambda \hat{F}(b_t)\tilde V_{\hat F,\hat G}(B_t-b_t,t+1)+\lambda (1-\hat{F}(b_t)) \tilde V_{\hat F,\hat G}(B_t,t+1),\\
    \tilde{H}_t&\coloneqq(v_t-b_t)F(b_t)+\lambda {F}(b_t)V_{F, G}(B_t-b_t,t+1)+\lambda (1-F(b_t)) V_{F,G}(B_t,t+1),\\
    H_t&\coloneqq(v_t-b_t^*)F(b_t^*)+\lambda {F}(b_t^*)V_{F, G}(B_t-b_t,t+1)+\lambda (1-F(b_t^*)) V_{F,G}(B_t,t+1).
\end{align*}

Recall that $R_1=\E_{\bm v\sim G^T} \left[\sum_{t=1}^T (H_t-\tilde H_t) \right]$. The following series of arguments holds. First, we have
\[ R_1\le \Biggl|{\E_{\bm v\sim G^T} \left[\sum_{t=1}^T (H_t-\hat H_t)\right]} \Biggr|+\E_{\bm v\sim G^T} \left[\sum_{t=1}^T|\tilde H_t-\hat H_t|\right]. \]
Followed from \Cref{lem:unknownG_diff_vv}, it holds that
\[ \Biggl|{\E_{\bm v\sim G^T}\left[\sum_{t=1}^T H_t-\hat H_t\right]}\Biggr|\le \left(\sqrt{\frac{1}{2} \ln \frac{4T}{\delta}} \frac{2 + \lambda}{(1 - \lambda)^2}+C_1 \right) \sum_{t=1}^T t^{-\frac{1}{2}}. \]

It suffices to bound $|\tilde H_t-\hat H_t|$. To do so we rewrite
\[ |\tilde H_t-\hat H_t|\le \Delta_1+\Delta_2+\Delta_3+\Delta_4+\Delta_5, \]
where
\begin{align*}
    \Delta_1&=|(v_t-b_t)\hat{F}(b_t)-(v_t-b_t)F(b_t)|,\\
    \Delta_2&=\lambda \hat{F}(b_t)|\tilde V_{\hat F,\hat G}(B_t-b_t,t+1)-V_{F, G}(B_t-b_t,t+1)|,\\
    \Delta_3&=|\lambda \hat{F}(b_t)V_{F, G}(B_t-b_t,t+1)-\lambda {F}(b_t)V_{F, G}(B_t-b_t,t+1)|,\\
    \Delta_4&=\lambda (1-\hat{F}(b_t)) |\tilde V_{\hat F,\hat G}(B_t,t+1)- V_{F,G}(B_t,t+1)| \text{ and}\\
    \Delta_5&=|\lambda (1-\hat F(b_t)) V_{F,G}(B_t,t+1)-\lambda (1- F(b_t)) V_{F,G}(B_t,t+1)|.
\end{align*}

\begin{itemize}
    \item To bound $\Delta_1$: using \Cref{lem:DKW}, it holds $\Delta_1\le\sqrt{\frac{1}{2} \ln \frac{4T}{\delta}} t^{-\frac{1}{2}} $.
    \item To bound $\Delta_2+\Delta_4$: using \Cref{lem:unknownG_diff_vv} and setting $C_1\leftarrow \frac{C_1}{\lambda}$, it holds that $\Delta_2+\Delta_4\le \lambda \left(\sqrt{\frac{1}{2} \ln \frac{4T}{\delta}} \frac{2 + \lambda}{(1 - \lambda)^2}+\frac{C_1}{\lambda} \right) t^{-\frac{1}{2}} $.
    \item To bound $\Delta_3$: since $V_{F, G}(\cdot,\cdot)\le \frac{1}{1-\lambda}$, then it holds that $\Delta_3\le \frac{\lambda}{1-\lambda}\sqrt{\frac{1}{2} \ln \frac{4T}{\delta}} t^{-\frac{1}{2}}$ .
    \item To bound $\Delta_5$: like the way we deal with $\Delta_3$, it holds that $\Delta_5\le \frac{\lambda}{1-\lambda}\sqrt{\frac{1}{2} \ln \frac{4T}{\delta}} t^{-\frac{1}{2}}$.
\end{itemize}

After summing them up, we have
\[ |\tilde H_t-\hat H_t|\le \left(\sqrt{\frac{1}{2} \ln \frac{4T}{\delta}} \frac{1 + 2\lambda}{(1 - \lambda)^2}+C_1 \right) t^{-\frac{1}{2}}. \]
Therefore, conditioning on the good event that the estimation succeeds (of both $F$ and $G$) for every $t$, $R_1$ is bounded by
\[ R_1\le \left(\sqrt{\frac{1}{2} \ln \frac{4T}{\delta}} \frac{3 + 3\lambda}{(1 - \lambda)^2}+2C_1 \right) \sum_{t=1}^T t^{-\frac{1}{2}}\le \left(\sqrt{\frac{1}{2} \ln \frac{4T}{\delta}} \frac{6(1 + \lambda)}{(1 - \lambda)^2}+4C_1 \right) \sqrt{T}. \]
We then apply the same techniques in \Cref{lem:r2} and \Cref{lem:final_regret} to yield
\[ \text{Regret} \leq \left(\sqrt{\frac{1}{2} \ln \frac{4T}{\delta}} \frac{6(1 + \lambda)}{(1 - \lambda)^2}+4C_1+1-\lambda \right) \sqrt{T}+\frac{ \lambda}{2} \log_\frac{1}{\lambda} \frac{T}{(1 - \lambda)^2}, \]
conditioning on the good event.

Now take $\delta = \frac{1}{T}$ and note that $\Pr[\text{bad event}] \leq \frac{1}{T}$. By using the trivial regret bound $T$ for the failure event, we obtain the desired bound
\[ \text{Regret} \leq \left(\sqrt{\frac{1}{2} \ln \frac{4T}{\delta}} \frac{6(1 + \lambda)}{(1 - \lambda)^2}+4C_1+1-\lambda \right) \sqrt{T}+\frac{ \lambda}{2} \log_\frac{1}{\lambda} \frac{T}{(1 - \lambda)^2}+1. \]

\section{Omitted proof in \Cref{subsect:censored}}
\subsection{Proof of \Cref{lem:TSN}}
Using \Cref{lem:convergencerate}, we have $ |{\Pr(\sup_b|\sqrt{2^n}(\hat{F}_n(1-b)-F(1-b))}|\ge r)-\Pr(\sup_b|B(1-b)|\ge r)|\le C_2 (1+r)^{-3}\ln^2 2^n (2^{-\frac{n}{6}})$. Since $2^n\le t\le 2^{n+1}-1$, it holds that:
\begin{gather*}
    |{\Pr(\sup_b|\sqrt{2^n}(\hat{F}_n(1-b)-F(1-b))}|\ge r)-\Pr(\sup_b|B(1-b)|\ge r)|\\ \le C_2 (1+r)^{-3}\ln^2 t \left(\frac{t}{2} \right)^{-\frac{1}{6}}.
\end{gather*}
Taking $M=C_2 2^{\frac{1}{6}}$ concludes the proof.

\subsection{Proof of \Cref{lem:Fbound}}
Before we proceed to bound the difference between $\hat{F}_n$ and $F$, we also need the following lemma to to characterize a property of Gaussian processes.
\begin{lemma}\label{lem:brownbridge}
    Let $\mathcal{B}(1-b)$ be a Gaussian process, we have
    \[ \Pr(\sup_b|\mathcal{B}(1-b)|\ge r)|\le C_3 e^{-C_4 r^2}, \]
    where $C_3$ and $C_4$ are constants.
\end{lemma}

For a Gaussian process, the tail satisfies Gaussian distribution. For a normal distribution, we have the following well-known inequality
\[
\int_x^\infty \frac{1}{\sqrt{2\pi}}e^{-\frac{t^2}{2}}\,\mathrm dt\le C_6 e^{-C_7 x^2},
\]
where $x\ge 0$ and $C_6$, $C_7$ are constants. For example, we can take $C_6=\frac{1}{2}e^{\frac{1}{2}}$ and $C_7=\frac{1}{2}$ in this inequality.

To bound for a certain Gaussian process, we rescale the random variable and the tail distribution also satisfies above property. So, there exists $C_3$ and $C_4$ to make \Cref{lem:brownbridge} hold.

Now we give a bound for the estimation of $\hat F$. First, we have
\[
\Pr(\sup_x |\sqrt{2^n}(\hat F(x)-F(x))|\ge r)\le \Pr(|\mathcal{B}|\ge r)+M(1+r)^{-3}\ln^2t \cdot t^{-\frac{1}{6}},
\]
We show that \Cref{lem:Fbound} establishes if we take $r=\max\{r_1, r_2\}$ where $r_1= \sqrt{\frac{1}{C_4}\ln\frac{4C_3\ln T}{\delta}}$ and $r_2=(4M\ln^3 T)^\frac{1}{3} t^{-\frac{1}{18}}\delta^{-\frac{1}{3}}$. Indeed, if $r \geq r_1$ then the first part on the right side is not greater than $\frac{\delta}{4\ln T}$ by \Cref{lem:brownbridge} and if $r \geq r_2$ then the second part is not greater than $\frac{\delta}{4\ln T}$. Take $\delta=T^{-\frac{5}{12}}$. It follows from simple comparison (between the rate of growth of $r_1$ and $r_2$) that there exists a constant $C_5$ for any $t$ and $T$ such that $\max\{r_1, r_2\} \leq C_5r_2$. Now, 
\[ \Pr\left(\sup_x \left|\sqrt\frac{t}{2}(\hat F(x)-F(x))\right|\ge C_5r_2\right) \leq \Pr(\sup_x |\sqrt{2^n}(\hat F(x)-F(x))|\ge C_5r_2) \leq \frac{1}{2T^{\frac{5}{12}}\ln T} \]
since $t \leq 2^{n+1}$. This concludes the proof of the lemma.

\subsection{Proof of \Cref{thm:censored_regret}}
Similar to the proof in the full-feedback case, let us first assume the knowledge of $G$. The regret of learning $G$ will be considered later. The proof of \Cref{thm:censored_regret} is short thanks to the tool-sets established by previous sections. As usual, we first condition on the good event that the estimation succeeds for every $t$. Then we add the contribution of the bad event to the regret in the final proof.

In order to bound the final regret, we need to bound the difference between $F$ and $\hat F_n$ and then apply the methodology used in the proof of \Cref{thm:regret_full_feedback}.

Dropping the first two rounds, for any $2^n\le t \le 2^{n+1}-1$, the estimation of $F$ is $\hat F_n$. Then using \Cref{lem:Fbound}, we obtain that, with probability at least $1-\frac{\delta}{2\ln T}$,
\[ \sup_x|\hat F_n(x)-F(x)|\le \sqrt{2}C_5(4M\ln^3T)^\frac{1}{3}t^{-\frac{5}{9}}T^{\frac{5}{36}}, \]
where $\delta = T^{-\frac{5}{12}}$. Note that we use the fact that if $1\le t \le T$, the algorithm updates for at most $\lfloor \ln T \rfloor$ times. Now, the new concentration bound (\Cref{lem:convergencerate}) effective changes $K, C$ and the convergence rate in \Cref{lem:r1}, \Cref{lem:r2} and \Cref{lem:final_regret}. Therefore, by substituting $\frac{C}{\sqrt t}$ with $\sqrt{2}C_5(4M\ln^3T)^\frac{1}{3}t^{-\frac{5}{9}}T^{\frac{5}{36}}\frac{1+\lambda}{(1-\lambda)^2}$ and $\frac{K}{\sqrt t}$ with $\sqrt{2}C_5(4M\ln^3T)^\frac{1}{3}t^{-\frac{5}{9}}T^{\frac{5}{36}}$ in the lemmas we obtain that conditioning on the good event (w.p. $1-\frac{\delta}{2}$) that the estimation of $F$ succeeds every time,
\[
R_1 \le 4C_1\sqrt{T}+\frac{9 \sqrt 2(1 + \lambda)}{2(1 - \lambda)^2}C_5(4M \ln^3T)^{\frac{1}{3}}T^{\frac{7}{12}},
\]
if \Cref{algo:censored_feedback} has the knowledge of $G$.

Next we add the estimation of $G$. Note that \Cref{lem:unknownG_diff_vv} decouples the regret of estimating $F$ and $G$, hence we may apply the same approach in the proof of \Cref{cor:unknownG} by taking $C=\sqrt{\frac{1}{2}\ln{\frac{4T}{\delta}}}\frac{1+\lambda}{(1-\lambda)^2}$ and $K=\sqrt{\frac{1}{2}\ln{\frac{4T}{\delta}}}$ and we have
\[
R_1 \le 4C_1\sqrt{T}+\frac{9 \sqrt 2(1 + \lambda)}{2(1 - \lambda)^2}C_5(4M \ln^3T)^{\frac{1}{3}}T^{\frac{7}{12}}+\left[\sqrt{\frac{1}{2}\ln{\frac{4T}{\delta}}}\frac{2(1+\lambda)}{(1-\lambda)^2} \right]\sqrt{T},
\]
conditioning on the good event (w.p. $1-\delta$) that the estimation of $F$ succeeds every time.

Finally, we use \Cref{lem:r2} and \Cref{lem:final_regret} to transform $R_1$ into the final regret, which yields
\begin{gather*}
    \text{Regret} \leq \frac{9 \sqrt 2(1 + \lambda)}{2(1 - \lambda)^2}C_5(4M \ln^3T)^{\frac{1}{3}}T^{\frac{7}{12}}+\left[\sqrt{\frac{1}{2}\ln4T^{\frac{17}{12}}}\frac{2(1+\lambda)}{(1-\lambda)^2} \right]\sqrt{T}\\+\left[(4C_1+1 -\lambda) \sqrt{T}+\frac{\lambda}{2} \log_\frac{1}{\lambda} \frac{T}{(1 - \lambda)^2}\right],
\end{gather*}
conditioning on the good event that the estimations of $F$ and $G$ succeeds for every $1 \leq t \leq T$. Note that $\Pr[\text{bad event}] \leq \frac{1}{T^{\frac{5}{12}}}$. By using the trivial regret bound $T$ for the failure event, we have
\begin{gather*}
    \text{Regret} \leq \left(\frac{9 \sqrt 2(1 + \lambda)}{2(1 - \lambda)^2}C_5(4M \ln^3T)^{\frac{1}{3}}+1\right)T^{\frac{7}{12}}\\+\left(\sqrt{\frac{1}{2}\ln{4T^{\frac{17}{12}}}}\frac{2(1+\lambda)}{(1-\lambda)^2}+4C_1+1 -\lambda\right) \sqrt{T}+\frac{ \lambda}{2} \log_\frac{1}{\lambda} \frac{T}{(1 - \lambda)^2}.
    \end{gather*}
This concludes the proof.